\documentclass[12pt,a4paper]{article}
\usepackage[T1]{fontenc}

\usepackage{amsmath,amsthm,amsfonts,amssymb}
\usepackage{enumerate}
\usepackage[bookmarks=false,colorlinks=true,linkcolor=black,citecolor=black,filecolor=black,urlcolor=black]{hyperref}
\usepackage{cleveref}
\usepackage{multirow}
\usepackage{booktabs}
\usepackage[table]{xcolor}
\usepackage{ulem}
\usepackage{algpseudocode}
\usepackage{soul}

\newtheorem{theorem}{Theorem}[section]
\newtheorem{definition}[theorem]{Definition}
\newtheorem{corollary}[theorem]{Corollary}
\newtheorem{proposition}[theorem]{Proposition}
\newtheorem{lemma}[theorem]{Lemma}

\newtheorem{remark}[theorem]{Remark}

\newtheorem{example}[theorem]{Example}

\newcommand{\Z}{{\mathbb{Z}}}

\newcommand{\zero}{{\mathbf{0}}}

\newcommand{\wt}{\operatorname{wt}}
\newcommand{\Aut}{\operatorname{Aut}}

\newcommand{\BH}{\operatorname{BH}}
\newcommand{\GH}{\operatorname{GH}}
\newcommand{\BHFP}{\operatorname{BHFP}}

\usepackage{url}
\usepackage{lineno}

\title{ Butson  full propelinear codes 
}

\newcommand{\abk}{\allowbreak}

\newcommand{\Sym}{\operatorname{Sym}}

\author{Jos\'e Andr\'es Armario\footnote{\url{E-mail: armario@us.es}}\\
\small Universidad de Sevilla, Sevilla, Spain\\
Ivan Bailera\footnote{\url{E-mail: ivan.bailera@uab.cat}}\\
        \small Universitat Aut\`onoma de Barcelona, Bellaterra, Spain\\
Ronan Egan\footnote{\url{E-mail: ronan.egan@nuigalway.ie}}\\
        \small National University of Ireland Galway, Galway, Ireland
}

\begin{document}

\maketitle
\begin{abstract}
In this paper we study Butson Hadamard matrices, and codes over finite rings coming from these matrices in logarithmic form, called BH-codes. We introduce a new morphism of Butson Hadamard matrices through a generalized Gray map on the matrices in logarithmic form, which is comparable to the morphism given in a recent note of \'{O} Cath\'{a}in and Swartz. That is, we show how, if given a Butson Hadamard matrix over the $k^{\rm th}$ roots of unity, we can construct a larger Butson matrix over the $\ell^{\rm th}$ roots of unity for any $\ell$ dividing $k$, provided that any prime $p$ dividing $k$ also divides $\ell$.

We prove that a $\Z_{p^s}$-additive code with $p$ a prime number is isomorphic as a group to a BH-code over $\Z_{p^s}$ and the image of this BH-code under the Gray map is a BH-code over $\Z_p$ (binary Hadamard code for $p=2$). Further, we investigate the inherent propelinear structure of these codes (and their images) when the Butson matrix is cocyclic. Some structural properties of these codes are studied and  examples  are provided.\\


 \end{abstract}

\section{Introduction}

Let $n$ and $k$  be positive integers, and $\zeta_k=\exp{(2\pi \sqrt{-1} /k)}$ be the complex $k^{\rm th}$ root of unity. We write $\langle \zeta_{k} \rangle = \{\zeta_{k}^{j}\}_{0 \leq j \leq k-1}$. Let $\Z_k$ be the ring of integers modulo $k$ with $k>1$, and denote by $\Z_k^n$ the set of $n$-tuples over $\Z_k$. We use bold notation ${\bf x} = [x_{1},\ldots,x_{n}] \in \Z_{k}^{n}$ to denote vectors (or codewords) in $\Z_{k}^{n}$. We denote the set of $n \times n$ matrices with entries in a set $X$ by $\mathcal{M}_{n}(X)$.

\subsection{Butson Hadamard matrices}

A \textit{Butson Hadamard (or simply Butson) matrix of order $n$ and phase $k$} is a matrix $H \in \mathcal{M}_{n}(\langle \zeta_{k} \rangle)$ such that $HH^*=nI_n$, where $I_n$ denotes the identity matrix of order $n$ and $H^*$ denotes the conjugate transpose of $H$. We write $\BH(n,k)$ for the set of such matrices. The simplest examples of Butson matrices are the Fourier matrices $F_n=[\zeta_n^{(i-1)(j-1)}]_{i,j=1}^n\in \BH(n,n)$. Hadamard matrices of order $n$, as they are usually defined, are the elements of $\BH(n,2)$. The phase and orthogonality of a $\BH(n,k)$ is preserved by multiplication on the left or right by a $n \times n$ monomial matrix with non-zero entries in the $k^{\rm th}$ roots of unity. For any pair of such monomial matrices $(P,Q)$ the operation defined by $H(P,Q) = PHQ^{*} = H'$ is an equivalence operation, and $H$ and $H'$ are said to be \textit{equivalent}. If $H = H'$, then $(P,Q)$ is an automorphism of $H$.

A Butson matrix $H\in \BH(n,k)$ is conveniently represented in logarithmic form, that is, the matrix $H=[\zeta_k^{\varphi_{i,j}}]_{i,j=1}^n$ is represented by the matrix $L(H)=[\varphi_{i,j}\mod k]_{i,j=1}^n$ with the convention that $L_{i,j}\in\Z_k$ for all $i,j\in\{1,\ldots,n\}.$
\begin{example} \label{formaLog}{\upshape
he following is a $\BH(4,8)$ matrix $H$, display in logarithmic form 
\[
L(H)=\left[\begin{array}{cccc}
   0 & 0 & 0 & 0    \\
   0 & 2 & 4 & 6     \\
   0 & 4 & 0 & 4 \\
    0 & 6 & 4 & 2 
\end{array}
\right]
\]}
\end{example}{\upshape}
Observe that the matrix above is in dephased form, that is, its first row and column are all $0$. Every matrix can be dephased by using equivalence operations. Throughout this paper all matrices are assumed to be dephased.
\begin{example}\label{krofp}{\upshape
Let $p$ be a prime number.  If $L(D)=[xy^{T}]_{x,y\in\Z_p^n} $ then $D$ is a $\BH(p^n,p)$. In fact $D$ is the $n$-fold Kronecker product of the Fourier matrix of order $p$. When $p=2$ this is the well known Sylvester Hadamard matrix of order $2^n$.}
\end{example}

Butson matrices have been subject to a considerable increase in interest recently for a variety of reasons. For one, a $\BH(n,k)$ exists for all $n$, (the Fourier matrix for example), but real Hadamard matrices, i.e., $\BH(n,2)$, exist when $n > 2$ only if $n \equiv 0 \mod 4$, and this condition is famously not yet known to be sufficient. A Butson morphism \cite{morphisms} is a map $\BH(n,k) \rightarrow \BH(m,\ell)$. This motives the study of Butson matrices even if real Hadamard matrices are the primary interest. In Section \ref{sec-Gray} we construct a morphism $\BH(n,k) \rightarrow \BH(nm,k/m)$ where $k = p_{1}^{e_{1}}\cdots p_{t}^{e_{t}}$ and $m = p_{1}^{e_{1}-1}\cdots p_{t}^{e_{t}-1}$, matching the parameters of the morphism discovered by \'{O} Cath\'{a}in and Swartz in \cite{OCS}. But their applications in applied sciences most strongly motivate their study. A $\BH(n,k)$ scaled by a factor of $1/\sqrt{n}$ is an orthonormal basis of $\mathbb{C}^{n}$. In any set of mutually unbiased bases (MUBs) which includes the standard basis, all other bases are necessarily of this form. MUBs have important applications in quantum physics, such as yielding optimal schemes of orthogonal quantum measurement  (see e.g., \cite{BZ}). Butson matrices also have applications in coding theory, as we discuss next.

\subsection{BH-codes and propelinear codes}
Interest in studying codes over finite rings increased significantly after it was proved in \cite{HKCSS94} that
certain notorious non-linear binary codes (such as the Preparata codes or the Kerdock codes), which had some of the properties of linear codes
were, in fact, the images of  codes over $\Z_4$ under a non-linear map (the Gray map). Codes constructed from Butson matrices \cite{GMO06,MW09,PR04,RP05} are a particular type of codes over a finite ring.
A code over $\Z_k$ (or $\Z_k$-code) of length $n$ is a nonempty subset $C$ of $\Z_k^n$. The elements of $C$ are called \textit{codewords}. The Hamming weight of a  vector ${\bf x}\in \Z_k$, denoted by $\wt_H({\bf x})$, is the number of nonzero coordinates of ${\bf x}$. The Hamming distance between two vectors ${\bf x}, {\bf y}\in \Z_k^n$, denoted by $d_{H}({\bf x},{\bf y})=\wt_H({\bf x-y})$, is the number of coordinates in which they differ. Given a minimum Hamming distance $d = \min_{{\bf x},{\bf y} \in C, {\bf x}\neq{\bf y}} d_{H}({\bf x},{\bf y})$ for a code $C$ of length $n$, we say $C$ is a $(n,|C|,d)$ code. Other distances functions are used, for instance, the Lee distance between two vectors ${\bf x}, {\bf y}\in \Z_k^n$ is $d_{L}({\bf x},{\bf y})=\wt_L({\bf x-y})$ where the Lee weight of a vector ${\bf z}=[z_1,\ldots,z_n]\in\Z_k^n$ is $\wt_L({\bf z})=\sum_{i=1}^n \wt_L(z_i)$ with $\wt_L(z_i)=\min \{z_i,k-z_i\}$.

Given $H\in \BH(n,k)$, we denote by $F_H$ the $\Z_k$-code of length $n$ consisting of the rows of $L(H)$, and by $C_H$ the  $\Z_k$-code defined as 
$C_H=\cup_{\alpha \in \Z_{k}}(F_H+\alpha {\bf 1})$ where ${\bf 1}$ denotes the all-one vector (and $\alpha{\bf 1}$ the all-$\alpha$ vector). 
We will write ${\bf 1}_n$ to denote the all-one vector of length $n$ when clarification is required. 
The code $C_H$ over $\Z_k$  is called a \textit{Butson Hadamard code} (briefly, BH-code). 

\begin{example} {\upshape Given $H\in \BH(4,8)$ of Example \ref{formaLog}. Then
$$
F_{H}=\{[0,  0, 0,  0],[0, 2,  4, 6],[0, 4, 0, 4], [0, 6, 4, 2]\},$$
$$C_{H}=\left\{\begin{array}{cccc}
   [0,  0, 0,  0], &     
   [0, 2,  4, 6], &
   [0, 4, 0, 4], &
    [0, 6, 4, 2], \\
    
    [1,  1, 1,  1], &     
   [1, 3,  5, 7], &
   [1, 5, 1, 5], &
    [1, 7, 5, 3], \\
    
    [2,  2, 2,  2], &     
   [2, 4,  6, 0], &
   [2, 6, 2, 6], &
    [2, 0, 6, 4], \\
    
    [3,  3, 3,  3], &     
   [3, 5,  7, 1], &
   [3, 7, 3, 7], &
    [3, 1, 7, 5], \\
    
    [4,  4, 4,  4], &     
   [4, 6,  0, 2], &
   [4, 0, 4, 0], &
    [4, 2, 0, 6], \\
    
    [5,  5, 5,  5], &     
   [5, 7,  1, 3], &
   [5, 1, 5, 1], &
    [5, 3, 1, 7], \\
    
     [6,  6, 6,  6], &     
   [6, 0,  2, 4], &
   [6, 2, 6, 2], &
    [6, 4, 2, 0], \\
    
     [7,  7, 7,  7], &     
   [7, 1,  3, 5], &
   [7, 3, 7, 3], &
    [7, 5, 3, 1]
    \end{array}
\right\}.
$$}
\end{example}

Assuming the Hamming metric, any isometry of $\mathbb{Z}_k^n$ is given by a coordinate permutation $\pi$ and $n$ permutations $\sigma_1,\ldots,\sigma_n$ of $\mathbb{Z}_k$. We denote by $\Aut(\mathbb{Z}_k^n)$ the group of all isometries of $\mathbb{Z}_k^n$:
$$\Aut(\mathbb{Z}_k^n)=\{(\sigma,\pi)\,\colon \,\sigma=(\sigma_1,\ldots,\sigma_n) \,\mbox{with}\, \sigma_i\in\,  \Sym \mathbb{Z}_k,\,\,\,
\pi\in{\cal S}_n\}
$$
where $\Sym \mathbb{Z}_k$ and ${\cal S}_n$ denote, respectively, the symmetric group of permutations on  $\mathbb{Z}_k$ and on the set   $\{1,\ldots,n\}$. The action of $(\sigma, \pi)$ is defined as
$$
(\sigma,\pi)({\bf v})=\sigma(\pi({\bf v})) \quad\mbox{for any} \,\,{\bf v}\in \mathbb{Z}_k^{n},
$$
and the group operation in $\Aut(\mathbb{Z}_k)$ is the composition
$$ 
(\sigma,\pi)\circ (\sigma',\pi')=((\sigma_1\circ \sigma'_{\pi^{-1}(1)},\ldots,\sigma_n\circ \sigma'_{\pi^{-1}(n)}),\pi\circ\pi')
$$
for all $(\sigma,\pi), (\sigma',\pi')\in \Aut(\mathbb{Z}_k)$.

\begin{definition}\label{PropDef}
A code $C$ of length $n$ over $\Z_k$ has a propelinear structure if for any codeword ${\bf x}\in C$ there exist $\pi_{\bf x}\in {\cal S}_n$ and $\sigma_{\bf x}=(\sigma_{{\bf x},1},\ldots,\sigma_{{\bf x},n})$ with $\sigma_{{\bf x},i}\in \Sym \Z_k$ satisfying:
 \begin{itemize}
     \item[(i)] 
     $(\sigma_{\bf x},\pi_{\bf x})(C)=C$ and 
    $(\sigma_{\bf x},\pi_{\bf x})({\bf 0})={\bf x}$,
     \item[(ii)] if ${\bf y}\in C$ and ${\bf z}=(\sigma_{\bf x},\pi_{\bf x})({\bf y})$, then 
     $(\sigma_{\bf z},\pi_{\bf z})=(\sigma_{\bf x},\pi_{\bf x})\circ(\sigma_{\bf y},\pi_{\bf y})$.
 \end{itemize}
 \end{definition}
 The propelinear structure was introduced in \cite{RBH89} for binary codes, and it was generalized in \cite{BMRS13} for $q$-ary codes.
 
For a code $C \subseteq \Z_{k}^{n}$, we denote by $\Aut(C)$ the group of all isometries of $\Z_{k}^{n}$ fixing the code $C$ and we call it the \textit{automorphism group} of the code $C$. A code $C$ over $\Z_k$ is called {\textit transitive} if $\Aut(C)$ acts transitively on its codewords, i.e., the code satisfies the property (i) of the above definition. 

Assuming that $C$ has a propelinear structure then a binary operation 
$\star$ can be defined as
$${\bf x}\star {\bf y}=(\sigma_{\bf x}, \pi_{\bf x})({\bf y})\quad \mbox{for any ${\bf x},{\bf y}\in C.$}$$
Therefore, $(C,\star)$ is a group, which is not abelian in general. This group structure is compatible with the Hamming distance, that is, $d_{H}({\bf x}\star {\bf u}, {\bf x}\star {\bf v})=d_{H}({\bf u},{\bf v})$ where ${\bf u},{\bf v}\in \Z_{k}^{n}$. 
The vector ${\bf 0}$ is always a codeword where $\pi_{\bf 0}=Id_n$ is the identity coordinate permutation and $\sigma_{{\bf 0},i }=Id_k$ is the identity permutation on $\Z_k$ for all $i\in\{1,\ldots,n\}$.
Hence, ${\bf 0}$ is the identity element in $C$ and 
$\pi_{{\bf x}^{-1}}=\pi_{\bf x}^{-1}$ and $\sigma_{{\bf x}^{-1},i}=\sigma^{-1}_{{\bf x},\pi_{\bf x}(i)}$
for all ${\bf x}\in C$ and for all $i\in\{1,\ldots,n\}$. We call $(C,\star)$ a \textit{propelinear code}. Henceforth we use $C$ instead of $(C,\star)$ if there is no confusion.

\begin{definition}\label{FullPropDef}
A \textit{full propelinear code} is a propelinear code $C$ such that for every ${\bf a} \in C$, $\sigma_{\bf a}({\bf x})={\bf a}+{\bf x}$ and $\pi_{\bf a}$ has not any fixed coordinate when ${\bf a}\neq \alpha {\bf 1}$ for $\alpha\in \Z_k$. Otherwise, $\pi_{\bf a}=Id_n$.
\end{definition}

\begin{remark}
{\upshape Every linear code is propelinear but not necessarily full.}
\end{remark}

A Butson Hadamard code, which is also full propelinear, is called a \textit{Butson Hadamard full propelinear code} (briefly, $\BHFP$-code). In the binary case, we have the Hadamard full propelinear codes, they were introduced in \cite{RS18} and their equivalence with Hadamard groups was proven. In the $q$-ary case, i.e., codes over the finite field $\mathbb{F}_{q}$ where $q$ is a prime power, the generalized Hadamard full propelinear codes were introduced in \cite{ABE19}. Their existence is shown to be equivalent to the existence of central relative $(n,q,n,n/q)$.
 
Propelinear codes are a topic of increasing interest in algebraic coding theory. The primary reason for this is that they offer one of the main benefits of linear codes, which is that they can be entirely described by a few generating codewords and group relations. However as the codes are not necessarily linear, they are not subject to all of the same minimum distance constraints as linear codes with the same number of codewords. Some propelinear codes may outperform comparable linear codes by having a larger minimum distance that any linear code of the same size, or by having more codewords than any linear code with a given minimum distance \cite{ABE19,HKCSS94}. In this paper we extend the work of the authors in \cite{ABE19} and describe the connection between cocyclic Butson Hadamard matrices and BHFP-codes.

\section{Constructing Butson Hadamard matrices and related codes}

Throughout this paper we study BH-codes over $\Z_{k}$. 
We have already introduced the Lee and Hamming distance between vectors ${\bf x}$ and ${\bf y}$.
We define other useful distance functions here. Initially, let $k = p^s$ for a prime $p$. The weight function $\wt^*(x)$ with $x\in\Z_{p^s}$ is defined by 
        $$\wt^*(x)=\left\{
        \begin{array}{ll}
             (p-1)p^{s-2} & x\neq kp^{s-1} \mod p^s,\,k\in\Z_p \\
        p^{s-1} & x= kp^{s-1}  \mod p^s, k\in\Z_p\setminus\{0\}\\
        0 & x=0 \mod p^s 
        \end{array}\right.
        $$
        
For $p=s=2$, this is the Lee weight. The corresponding distance $d^*$ on $\Z_{p^s}^n$ is defined as follows:
\begin{equation}\label{d-star}
d^*({\bf x},{\bf y})=\sum_{i=1}^n \wt^*(y_i-x_i),
\end{equation}
where ${\bf x}=[x_1,\ldots,x_n]$ and 
${\bf y}=[y_1,\ldots,y_n]$ in $\Z_{p^s}^n$. More generally, let $k = mp^{s}$ for $m$ coprime to $p$. Any $x \in \Z_{k}$ may be written uniquely in the form $x = ap^{s} + bm \mod k$ where $0 \leq a \leq m-1$ and $0 \leq b \leq p^{s}-1$.
Define the weight function $\wt^{\dag}(x)$ on $\Z_{k}$ by 
        $$\wt^{\dag}(x)=\left\{
        \begin{array}{ll}
             \wt^{*}(b) & a = 0 \\
        p^{s-1} & a \neq 0.
        \end{array}\right.
        $$

The definition of the weight function here is consistent with the homogeneous metric introduced in \cite{CH97}.
The corresponding distance $d^{\dag}$ on $\Z_{mp^s}^n$ is defined as follows:
\begin{equation}\label{d-dag}
d^{\dag}({\bf x},{\bf y})=\sum_{i=1}^n \wt^{\dag}(y_i-x_i),
\end{equation}
where ${\bf x}=[x_1,\ldots,x_n]$ and 
${\bf y}=[y_1,\ldots,y_n]$ in $\Z_{mp^s}^n$. 

Given $H\in \BH(n,k)$, recall that $F_H$ is the $\Z_k$-code of length $n$ consisting of the rows of $L(H)$, and $C_H=\cup_{\alpha\in \Z_k}(F_H+\alpha {\bf 1})$. Let  $r_H^i(l)$ be the number of repetitions of $l\in \Z_{k}$  in the i-$th$ row of $L(H)$ and $r_H(l)= \displaystyle\max_{2\leq i\leq n} r_H^i (l)$.

For $k=p^s$, Lemma 3.1 of \cite{LOS20} gives a pattern that any row of $L(H)$ has to follow. That is, any row ${\bf x}$ has to be a permutation of the vector $({\bf u},r_1 {\bf 1}+{\bf u},\ldots,r_{t-1}{\bf 1}+{\bf u})$ where ${\bf u}=[0,p^{s-1}, 2p^{s-1},\ldots,,(p-1)p^{s-1}]$, $r_i\in \Z_{p^{s-1}},$ for $1\leq i\leq t-1$ with 
$t=\frac{n}{p}$. Therefore,
$$r_H(l)\leq\left\{ 
\begin{array}{cl}
     \frac{n}{p} & l=h p^{s-1}\,\mbox{where $h\in\Z_p$} \\
     \frac{n}{p}-1& \mbox{Otherwise.} 
\end{array}
\right.$$
As a consequence, $n-\frac{n}{p}$ is an upper bound for the minimum Hamming distance of $F_H$ when $k=p^s$. Furthermore, the minimum Hamming distance of both codes, $F_H$ and $C_H$, is the same in this case. 

In \cite{PR04,RP05}, the authors prove that if $n=p^{sm}$ and $k=p^s$ then the {
minimum} Hamming  distance of $F_H$ is   $n-\frac{n}{p}$ and the {
minimum} Lee distance is given by
$$d_L=\left\{\begin{array}{cc}
          2^{m+s-2},   & p=2  \\
          \frac{p^{s(m+1)-2}}{4}
    (p^2-1),   & p>2 \,\,\,\mbox{prime;}
        \end{array}
            \right.$$
where $H$ is the Butson matrix of Theorem \ref{conbugac} and $m=t_1-1$ for $t_1>0$ and $t_2=\ldots=t_s=0$.

Finally, Theorem 5.4 of \cite{GMO06} claims that for any pair $(n,k)$ such that $\BH(n,k)\neq \emptyset$, if $H\in \BH(n,k)$ then the code obtained by deleting the first coordinate in $F_H$ has parameters $(n-1,n,\gamma n)$ meeting the Plotkin bound over Frobenius rings where $\gamma$ is the average homogeneous weight over $\Z_k$.

\subsection{A Fourier type construction and simplex codes}
In what follows, we describe a method to construct Butson matrices of order $n=p^{st_1+(s-1)t_2+\ldots + t_s-s} $ and phase $k=p^s$, where $p$ is a prime. 
Let $s$ be a positive integer, $t_1,t_2,\ldots, t_s$ be nonnegative integers with $t_1\geq 1$, and $A^{1,0,\ldots,0}=[0]$. The matrix 
$A^{t_1,t_2,\ldots, t_s}$, where ${\bf p^{i-1}}$ denotes the all-$p^{i-1}$ vector, is defined recursively according to the following algorithm, where initially, $(t_1',t_2',\ldots,t_s')=(1,0,\ldots,0)$.\newline

\begin{algorithmic}
\For{i=1 until s}
   \While{ $t_i'<t_i$}
     \State{$A\leftarrow A^{t_1',\ldots,t_s'}$}
      \State{$t_i'\leftarrow t_i'+1$}
      \State{$A^{t_1',\ldots,t_s'}\leftarrow A_i=\left[\begin{array}{cccc}
          A & A & \ldots & A \\
         0\cdot {\bf p^{i-1}} & 1\cdot {\bf p^{i-1}} & \ldots &  
         (p^{s-i+1}-1)\cdot {\bf p^{i-1}}
     \end{array}\right]$}
    \EndWhile
\EndFor    \newline
\end{algorithmic}

By construction, it is clear that $A^{t_1,t_2,\ldots, t_s}$ is a $(t_1+t_2+\ldots+t_s)\times (p^{st_1+(s-1)t_2+\ldots + t_s-s})$ matrix. This is a generalization of the construction of   \cite{PR04} as we will point out in Corollary \ref{Ashannew}. 

\begin{example}{\upshape
For $p=2$ and $s=3$. We have $A^{1,1,0}=\left[\begin{array}{cccc}
     0 & 0 & 0 & 0  \\
     0 & 2 & 4 & 6
    \end{array}
\right]$ and $A^{1,1,1}=\left[\begin{array}{cccccccc}
     0 & 0 & 0 & 0 & 0 & 0 & 0 & 0\\
     0 & 2 & 4 & 6 &  0 & 2 & 4 & 6\\
      0 & 0 & 0 & 0 & 4 & 4 & 4 & 4
    \end{array}
\right].$}

\end{example}

Given ${\bf x} \in \Z_{k}^{n}$, the \textit{order} of ${\bf x}$ is the smallest positive integer $m$ such that $m{\bf x} = {\bf 0}$ over $\Z_{k}$.

\begin{lemma}\label{lemmaayudaconsbut}
Let  $k=p^s$ and  ${\bf u_i}=[\begin{array}{cccc}
   0\cdot {\bf p^{i-1}}, & 1\cdot {\bf p^{i-1}}, & \ldots, &  
         (p^{s-i+1}-1)\cdot {\bf p^{i-1}}  
\end{array}]\in \Z_{p^s}^{p^{s-i+1}}$ where $1\leq i\leq s$. Then
\begin{itemize}
\item $\displaystyle\sum_{j=0}^{p^{s-i+1}-1}\zeta_k^{jp^{i-1}}=0$, for all $1 \leq i \leq s$.
\item The order of ${\bf u_i}$ is $p^{s-i+1}$.
\item If $\gcd(m,p^{s-i+1})=1$ then  $[m{\bf u_i}]_j=[{\bf u_i}]_{\pi(j)}$ where $\pi\in{\cal S}_{p^{s-i+1}}$.
\item If $\gcd(m,p^{s-i+1})=p^l$ then  $[m{\bf u_i}]_j=[\frac{m}{p^l} {\bf u_{i+l}}]_{\pi(j\mod p^{s-(i+l)+1})}$ where $\pi\in{\cal S}_{p^{s-(i+l)+1}}$.  
\end{itemize}
\end{lemma}
\begin{proof}
It is a straightforward comprobation. 
\end{proof}

\begin{theorem}\label{conbugac}
Let $n=p^{st_1+(s-1)t_2+\ldots + t_s-s} $ and $L(H)$ be the $n\times n$ matrix whose rows are  the $n$ possible linear combinations (with coefficients in $\Z_{p^s}$) of the rows of $A^{t_1,t_2,\ldots, t_s}$. Then,
$H\in \BH(n,p^s)$.
\end{theorem}
\begin{proof}
By construction, the difference between two distinct rows of $L(H)$ is a linear combination (with coefficients in  $\Z_{p^s}$) of the rows of $A^{t_1,t_2,\ldots, t_s}$. Hence, it is a row of $L(H)$.

Therefore, proving $HH^*=nI_n$ reduces to proving that every row sum  of $H$ is $0$. For the rows of $H$ corresponding to multiples of  the rows of $A^{t_1,t_2,\ldots, t_s}$, this holds as a consequence of Lemma \ref{lemmaayudaconsbut}. Finally, the proof for the rows of $H$ corresponding to a linear combination of the rows of $A^{t_1,t_2,\ldots, t_s}$ is by a simple induction. 
\end{proof}

We provide some examples of Butson matrices coming from Theorem \ref{conbugac}.
\begin{example} \label{ejemplosbhadd} {\upshape
Let $p=2$ and $s=3$. For $t_1=1,\,t_2=1,\,t_3=0$ then $L(H)$ is the matrix given in Example \ref{formaLog}. For $t_1=1,\,t_2=1,\,t_3=1$ then $H\in \BH(8,8)$ where
\[
L(H)=\left[
\begin{array}{cccccccc}
     0 & 0 & 0 & 0 & 0 & 0 & 0 & 0\\
     0 & 2 & 4 & 6 &  0 & 2 & 4 & 6\\
      0 & 4 & 0 & 4 &  0 & 4 & 0 & 4\\
       0 & 6 & 4 & 2 & 0 & 6 & 4 & 2 \\
            0 & 0 & 0 & 0 & 4 & 4 & 4 & 4\\
         0 & 2 & 4 & 6 & 4 & 6 & 0 & 2  \\
         0 & 4 & 0 & 4 &  4 & 0 & 4 & 0\\
        0 & 6 & 4 & 2 & 4 & 2 & 0 & 6   
         
\end{array}
\right]
\]}
\end{example}
\begin{remark}{\upshape
Let $L(H)$ be the matrix of Example \ref{ejemplosbhadd} for $t_1=1,\,t_2=1,\,t_3=1$. Then
$L(H)=L(F_2\otimes F_4) $ where we have used that $F_2\otimes F_4\in \BH(8,8)$ by means of $\zeta_2=\zeta_8^4$ and $\zeta_4=\zeta_8^2.$}
\end{remark}

In general we have the following.
\begin{proposition}\label{adpkfm}
Let $n=p^{st_1+(s-1)t_2+\ldots + t_s-s} $ and $L(H)$ be the $n\times n$ matrix of Theorem \ref{conbugac}. Then, $H$ is equivalent to 
\[
(F_{p})^{t_s}\,\otimes \,(F_{p^2})^{t_{s-1}}\,\otimes\,\ldots\,\otimes \,(F_{p^{s-1}})^{t_2}\,\otimes\,(F_{p^s})^{t_1-1}
\]
where 
$F_{p^{s-j}}$ denotes  the Fourier matrix of order $p^{s-j}$ embedded in $\BH(p^{s-j},p^s)$ using that $\zeta_{p^{s-j}}=\zeta_{p^s}^{p^j}$, and $(M)^r$ denotes the $r$-fold Kronecker  product of the matrix $M$.
\end{proposition}
\begin{proof}
The proof is by induction. The case $t_{1},t_{2},\ldots, t_{s} = 1,0,\ldots,0$ is trivial, so consider the case $t_1=2$ and $t_2=\ldots=t_s=0$. It is clear that
$L(H)=L(F_{p^s})$ since $A^{2,0\ldots,0} = \left[\begin{array}{cccc}
     0 & 0 & \cdots & 0  \\
     0 & 1 & \cdots & p^{s}-1
    \end{array}
\right]$. For the next step of the induction, we  assume that
 $t_{i+1}=\ldots =t_s=0$ and $L(H)=L( (F_{p^{s-(i-1)}})^{t_{i}'}\,\otimes\,(F_{p^{s-(i-1)}})^{t_{i-1}}\,\otimes\,\ldots\,\otimes \,(F_{p^{s-1}})^{t_2}\,\otimes\,(F_{p^s})^{t_1-1})$. Now,  we have to distinguish two possibilities:
\begin{itemize}
    \item $t_i'<t_i$; then let $t_{i}' \leftarrow t_{i}'+1$ and $t_{i+1}=\ldots =t_s=0.$ All the possible linear combinations of the rows of $A^{t_1,\ldots,t_{i-1},t_i'+1,0\ldots,0}$ are the rows of $B=L(F_{p^{s-(i-1)}}\otimes H)$. 
    \item $t_i'=t_i$; then take $t_{i+1}=1$ with $t_{i+2}=\ldots =t_s=0.$ Proceeding in a similar way, the result holds.
\end{itemize}
\end{proof}

It is clear now that this construction is not new, in the sense that it does not produce any Butson matrices not already known. However this perspective gives us new insights into the related BH-codes. 
\begin{remark}
{\upshape For $t_1\neq 0$, $t_2=\ldots=t_s=0$  and $p=2$, the code generated with the rows of $A^{t_1,0,\ldots,0}$ is  a $\Z_{p^s}$-simplex code of type $\alpha$ (see \cite[Definition 4.1]{PR04}).  Furthermore, this code is self-orthogonal if $s=2$.}
\end{remark}

\begin{corollary}\label{Ashannew}
 A simplex code of type $\alpha$ over $\Z_{2^s}$ of length $2^{sm}$ (see \cite{PR04}) and the code whose codewords are the rows of $ L((F_{2^s})^m)$ 
 are the same.  Therefore $M_\psi$, the cocyclic $\BH(2^{sm},2^s)$ of \cite[Theorem 5.1, ii)]{PR04} is equivalent to $(F_{2^s})^m$. Similarly, when $p>2$ prime, the analogous classifying result for the cocyclic $\BH(p^{sm},p^s)$ of \cite[Proposition 3.1, ii)]{RP05} holds.
\end{corollary}

\begin{proof}
Attending to the Remark above, a simplex code of type $\alpha$ over $\Z_{2^s}$ of length $n=2^{st_1}$ is exactly the code $F_H$ where $H$ is the $n\times n$ matrix of  Theorem \ref{conbugac}. Applying Proposition \ref{adpkfm}, the results follows.
\end{proof}

The classifying result above follows also as a consequence of \cite[Theorem 13]{MW09}.




A nonempty subset ${\cal C}$ of $\Z_{p^s}^n$ is \textit{a $\Z_{p^s}$-additive code} if it is a subgroup of $\Z_{p^s}^n$ (i.e., a $\Z_{p^s}$-module). Clearly, given \textit{a $\Z_{p^s}$-additive code}, ${\cal C}$, of length $n$ there exist some non negative integers $t_1,\ldots,t_s$ such that ${\cal C}$ is isomorphic (as an abelian group) to $\Z_{p^s}^{t_1}\times \Z_{p^{s-1}}^{t_2}\times \ldots\times \Z_p^{t_s}$. Thus,  ${\cal C}$ is said to be of type $(n;t_1,\ldots,t_s)$.  Note that $|{\cal C}|=p^{st_1}p^{(s-1)t_2}\ldots p^{t_s}$ since there are $t_1$ (generators) codewords of order $p^s$, $t_2$ of order $p^{s-1}$ and so on.

\begin{remark}\label{Defofaddcodes}{\upshape Let $t_1,\ldots,t_s$ be non negative integers and
taking $A^{1,0,\ldots,0}=[1]$ instead of $[0]$, the method described at the beginning of this section provides $A^{t_1,t_2,\ldots, t_s}$ as a generator matrix for a $\Z_{p^s}$-additive code of type $(n;t_1,\ldots,t_s)$ where $n=p^{st_1+(s-1)t_2+\ldots + t_s-s}$. The description of 
recursive constructions of these matrices are in \cite{FVV19,Kro01,Kro07} for $p=2$. The case $p\neq 2$ has been studied in \cite{SWK19}.
We will denote the codes associated to these matrices by ${\cal H}^{t_1,\ldots,t_s}$. Let us point out that ${\cal H}^{0,t_2,\ldots,t_s}\subset {\cal H}^{1,t_2,\ldots,t_s}$.}
\end{remark}

Now, we establish  the following result.

\begin{proposition}\label{ACareBH} For $t_1>0$,
every ${\cal H}^{t_1,\ldots,t_s}$ is a $BH$-code where the Butson Hadamard matrix is a Kronecker product of Fourier matrices. 
\end{proposition}
\begin{proof}  Let $C_H$ be the BH-code associated to $H$ of Theorem \ref{conbugac}. It is clear that $C_H$ is equivalent to  ${\cal H}^{t_1,\ldots,t_s}$. Now, the result follows from Proposition \ref{adpkfm}.
\end{proof}



  The following is an example of a  BH-code which is not additive.
\begin{example}\label{ejbhcnl}{\upshape
Let $H\in\BH(8,4)$ with
$$L(H)=\left(\begin{array}{cccccccc}
0&0&0&0 & 0&0&0&0\\
0&1&3&0 & 2&3&1&2\\
0&3&2&1  & 0&3&2&1\\
0&0&1&1 & 2&2&3&3\\
0&2&0&2 & 0&2&0&2\\
0&3&3&2 & 2&1&1&0\\
0&1&2&3 & 0&1&2&3\\
0&2&1&3 & 2&0&3&1
\end{array}\right).$$
$C_H$ is not $\Z_{2^2}$-additive since the double of the second row is not a codeword.}

\end{example}

Now, we can state that, in a certain sense, the class of BH-codes encompasses strictly  the class of $\Z_{p^s}$-additive codes. Since any $\Z_{p^s}$-additive code is always of type $(n;t_1,\ldots,t_s)$ for some non negatives integers $t_1,\ldots,t_s$ and ${\cal H}^{t_1,\ldots,t_s}\in \BH(p^{st_1+(s-1)t_2+\ldots + t_s-s},p^s)$, assuming that $t_1>0.$

\subsection{Generalized Gray map}\label{sec-Gray}
The Gray map is a function from $\Z_{4}$ to $\Z_{2}^{2}$ which is typically used to form binary codes from $\Z_4$-codes. In what follows, we introduce a generalized Gray map $\Phi_{p}$ from $\Z_{p^s}$ to $\Z_p^{p^{s-1}}$, and extend this to a yet more general function $\Psi_{p}$ from $\Z_{mp^{s}}$ to $\Z_{mp}^{p^{s-1}}$. For $k = p_{1}^{e_{1}}\cdots p_{t}^{e_{t}}$, and $\ell = p_{1}\cdots p_{t}$ the composition $\Psi_{p_{t}}\cdots \Psi_{p_{1}}$ is a function from $\Z_{k}$ to $\Z_{\ell}^{k/\ell}$. From this function we construct a morphism $\BH(n,k) \rightarrow \BH(nk/\ell,\ell)$. Where ${\bf x} = [x_{1},\ldots,x_{n}] \in \Z_{k}^{n}$ and $\varphi$ is any function with domain $\Z_{k}$, we will write $\varphi({\bf x}) = [\varphi(x_{1}),\ldots,\varphi(x_{n})]$. Further, we write $\varphi(\mathcal{C}) = \{\varphi({\bf c}) \; : \; {\bf c} \in \mathcal{C}\}$ where $\mathcal{C} \subseteq \Z_{k}^{n}$.


We consider the elements of $\Z_p^{s-1}$ to be ordered in increasing lexicographic order. We denote by $D$ the $\BH(p^{s-1},p)$  matrix defined in Example \ref{krofp} and label the rows of $L(D)$ in the order $0,1,\ldots,p^{s-1}-1$. Let $[L(D)]_i$ denotes the row of $L(D)$ labeled by $i$. Then we let  $\Phi_{p} : \Z_{p^s} \rightarrow \Z_p^{p^{s-1}}$ be the map defined by
\[
\Phi_{p}(x)= [L(D)]_{b}+a{\bf 1},\quad x=ap^{s-1}+b.
\]
Let us observe that for $p=2$, $\Phi_{p}$ is the well-known Carlet's map \cite{Car98} and for $p>2$, $\Phi_{p}$ is of type $\varphi$ given in \cite{SWK19}. For what remains of this section we write $\Phi = \Phi_{p}$ for brevity unless there is some confusion.

\begin{proposition}[\cite{SWK19}]\label{isophi}
The entrywise application of $\Phi$ is an isometric embedding of $(\Z_{p^s}^n, d^*)$ into $(\Z_p^{p^{s-1} n},d_H)$. Furthermore, if ${\cal C}$ is a code with parameters $(n,M,d^*)$ over $\Z_{p^s}$, then the image code $C=\Phi({\cal C})$ is a code with parameters $(p^{s-1}n,M,d_H)$ over $\Z_p$.
\end{proposition}
\begin{lemma}\label{alpha-diff}
Let $x,y \in \Z_{p^s}$. Then $\Phi(x-y)=\Phi(x)-\Phi(y)+ \alpha {\bf 1}$ where $\alpha \in \{0,p-1\}$.
\end{lemma}

\begin{proof}
Let $x = a_{1}p^{s-1} + b_{1}$ and $y = a_{2}p^{s-1}+b_{2}$. Then
\[
x-y = \begin{cases} 
(a_{1}-a_{2})p^{s-1} + (b_{1}-b_{2}), \quad \text{ if }~ b_{1} \geq b_{2} \\
(a_{1}-a_{2}-1)p^{s-1} + (b_{1}-b_{2}), \quad \text{ if } b_{1} < b_{2}.
\end{cases}
\]
Further, by the linearity of the inner product $vw^{T}$ and the definition $L(D)=[vw^{T}]_{v,w\in\Z_p^n}$ it follows that $\Phi(b_{1}-b_{2}  \mod p^{s-1})=[L(D)]_{b_{1}-b_{2}}=[L(D)]_{b_{1}}-[L(D)]_{b_{2}}=\Phi(b_1)-\Phi(b_2)$. Thus $\Phi(x-y)=\Phi(x)-\Phi(y)+ \alpha {\bf 1}$ where $\alpha = 0$ if $b_{1} \geq b_{2}$, and $\alpha = p-1$ otherwise.
\end{proof}

Given $H\in \mathcal{M}_{n}(\langle \zeta_{p^s} \rangle)$, we write $L(H^{\Phi})$ for the entrywise application of $\Phi$ to
\[\left[\begin{array}{c}
     L(H)  \\
     L(H)+J\\
     L(H)+2J\\
     \vdots\\
     L(H)+(p^{s-1}-1)J
\end{array}\right].\]
Then $H^{\Phi}$ is the corresponding matrix in $\mathcal{M}_{np^{s-1}}(\langle \zeta_{p} \rangle)$.

\begin{theorem}
If $H\in \BH(n,p^s)$, then $H^{\Phi}\in \BH(np^{s-1},p).$
\end{theorem}

\begin{proof}
Observe that $H^{\Phi}$ is Butson Hadamard over $\langle \zeta_{p} \rangle$ if, for all $i\neq j$, the sequence of differences $[L(H^{\Phi})]_{i,l}-[L(H^{\Phi})]_{j,l}, \,0\leq l\leq n\cdot p^{s-1}-1$ contains each element of $\Z_p$ equally often. First note that 
for all $i\neq j$, the sequence of differences $[L(H)]_{i,l}-[L(H)]_{j,l}, \,0\leq l\leq n-1$ contains each element of the form $ap^{s-1}$ equally often for $a=0,\ldots,p-1.$  This is a consequence of $\zeta_k^{ap^{s-1}}$ being a $p^{\rm th}$ root of unity. By Lemma \ref{alpha-diff}, if $x-y = ap^{s-1}$ then $\Phi(x-y)=\Phi(x)-\Phi(y)$. Since $\Phi(ap^{s-1})=a{\bf 1}$ for $a\in \Z_{p}$, it follows that if the set of differences $[L(H)]_{i,l}-[L(H)]_{j,l}$ contains $m$ repetitions of each element of the form $ap^{s-1}$, then the set of corresponding differences in $[L(H^{\Phi})]_{i,l}-[L(H^{\Phi})]_{j,l}$ contains $mp^{s-1}$ repetitions of each element of $\Z_{p}$. Finally, if $x-y \not\equiv 0 \mod p^{s-1}$, then $\Phi(x)-\Phi(y) = \Phi(x-y) + \alpha{\bf 1}$ for some $\alpha$, where $x-y = ap^{s-1}+b$ and $b \neq 0$. Thus $\Phi(x-y) = a{\bf 1} + [L(D)]_{b}$ which contains every element of $\Z_p$ exactly $p^{s-2}$-times, and so too does $\Phi(x)-\Phi(y)$.
\end{proof}

\begin{corollary}\label{coroltmbh}
The image of any BH-code over $\Z_{p^s}$ of length $n$  by $\Phi$  is  a BH-code over $\Z_p$ of length $n\cdot p^{s-1}$ and minimum Hamming distance $d_H=np^{s-2}(p-1).$
\end{corollary}

\begin{remark}
{\upshape Let us point out that  Theorem 1 of \cite{FVV19} is a particular case of Corollary \ref{coroltmbh} (when the BH-code is of type ${\cal H}^{t_1,\ldots,t_s}$ and $p=2$).}
\end{remark}
\begin{proposition}
Any BH-code ${ C_H}$ of length $n$ over $\Z_{p^s}$ has minimum distance $d^*=np^{s-2}(p-1)$.
\end{proposition}
\begin{proof}
Taking into account that $\BH(n,p)=\GH(p,n/p)$ where $\GH(p,n/p)$ denotes the set of generalized Hadamard matrices of order $n$ over  $\mathbb{F}_{p}$ (see \cite[Lemma 2.2]{EFO15}). Thus, $C_{H^\Phi}=\Phi(C_H)$ is a generalized Hadamard code as well since $ H^\Phi\in \BH(p^{s-1}n,p)$. The minimum Hamming distance of these codes is well known to be $np^{s-2}(p-1)$. The fact that $\Phi$ is an isometric embedding (Proposition \ref{isophi}) concludes the proof.  
\end{proof}

Now let $k = mp^{s}$ where $p$ does not divide $m$ and recall that every element $x \in \mathbb{Z}_{k}$ can be written uniquely as $x = ap^{s} + bm \mod k$ for some $0 \leq a \leq m-1$ and $0 \leq b \leq p^{s}-1$. Then let
\[
\Psi_{p}(ap^{s} + bm) = m \Phi_{p}(b) + ap{\bf 1}
\]
define a map $\mathbb{Z}_{k} \rightarrow \mathbb{Z}_{mp}^{p^{s-1}}$.

\begin{proposition}\label{isophi2} 
The entrywise application of $\Psi_{p}$ is an isometric embedding of $(\Z_{mp^s}^n, d^{\dag})$ into $(\Z_{mp}^{p^{s-1} n},d_H)$. Furthermore, if ${\cal C}$ is a code with parameters $(n,M,d^{\dag})$ over $\Z_{mp^s}$, then the image code $C=\Psi_{p}({\cal C})$ is a code with parameters $(p^{s-1}n,M,d_H)$ over $\Z_{mp}$.
\end{proposition}
\begin{proof}
This follows from a straight forward extension of Proposition \ref{isophi}. 
\end{proof}

Given $H\in \mathcal{M}_{n}(\langle \zeta_{k} \rangle)$ where $k = p^{s}m$, we write $L(H^{\Psi_{p}})$ for the entrywise application of $\Psi_{p}$ to
\[\left[\begin{array}{c}
     L(H)  \\
     L(H)+mJ\\
     L(H)+2mJ\\
     \vdots\\
     L(H)+(p^{s-1}-1)mJ
\end{array}\right].\]
Then $H^{\Psi_{p}}$ is the corresponding matrix in $\mathcal{M}_{np^{s-1}}(\langle \zeta_{pm} \rangle)$. We will devote the rest of this section to a proof of the following.

\begin{theorem}\label{general morphism}
If $H\in \BH(n,k)$ where $k = p^{s}m$, then $H^{\Psi_{p}}\in \BH(np^{s-1},pm).$
\end{theorem}

Repeated application of $\Psi_{p}$ for all primes $p$ dividing $k$ gives the following.

\begin{corollary}
If there exists a $\BH(n,k)$ where $k = p_{1}^{s_{1}}\cdots p_{r}^{s_{r}}$, then there exists a $\BH(nk/\ell,\ell)$ where $\ell = p_{1}\cdots p_{r}$.
\end{corollary}

Before we can prove Theorem \ref{general morphism}, we will need to establish some preliminary results. Hereafter we fix a prime $p$ and let $\Psi = \Psi_{p}$.

\begin{lemma}\label{Psi decomp}
For all $0 \leq x,y < k = mp^{s}$, $\Psi(x-y) = \Psi(x)-\Psi(y)+m\alpha{\bf 1}$ where $\alpha \in \{0,p-1\}$.
\end{lemma}


\begin{proof}
Let $x = ap^s + bm$ and $y = cp^s + dm$. Observe that $\Psi(x-y) = (a-c)p{\bf{1}} + m\Phi(b-d)$. By Lemma \ref{alpha-diff}, $\Phi(b-d) = \Phi(b) - \Phi(d) + \alpha{\bf 1}$ where $\alpha \in \{0,p-1\}$. The result follows.
\end{proof}

\begin{lemma}\label{x-psi lemma}
Let $z \neq fp^{s-1}$ for any $0 \leq f \leq mp-1$. Then $\sum_{i=1}^{p^{s-1}}\omega^{\Psi(z)_{i}} = 0$ where $\omega$ is a primitive $k^{\rm th}$ root of unity. Otherwise, $\Psi(z) = f{\bf 1}$, and $\sum_{i=1}^{p^{s-1}}\omega^{\Psi(z)_{i}} = p^{s-1}\omega^{f}$.
\end{lemma}
\begin{proof}
First suppose that $z \neq fp^{s-1}$. Observe that $\Psi(z) = m[L(D)]_{j} + \alpha{\bf 1}$ for some $\alpha \in \mathbb{Z}_{pm}$ and $j \neq 0$. Then $\sum_{i=0}^{p^{s-1}-1}\omega^{\Psi(z)_{i}} = \sum_{i=1}^{p^{s-1}}\omega^{[L(D)]_{j,i}+\alpha} = \omega^{\alpha}\sum_{i=0}^{p^{s-1}-1}\omega^{[L(D)]_{j,i}} = 0$.

Now suppose that $z = fp^{s-1}$. Then $f = gm+hp \mod mp$ where $0 \leq g \leq p-1$ and $0 \leq h \leq m-1$. Thus $fp^{s-1} = hp^{s} + gmp^{s-1} \mod p^{s}m$. It follows that $\Psi(z) = hp{\bf 1} + m\Phi(gp^{s-1}) = hp{\bf 1} + gm{\bf 1} = f{\bf 1}$.
\end{proof}


\begin{corollary}\label{x-y corollary}
If $x = fp^{s-1}$ and $y \neq 0 \mod p^{s-1}$, then $\Psi(x-y) = \Psi(x)-\Psi(y) + m(p-1){\bf 1}$.  Consequently, for any multiset $X$ of elements of $\mathbb{Z}_{k}$ such that $x \in X$ only if $x = fp^{s-1}$,and for any $y \neq 0 \mod p^{s-1}$,
then $\sum_{x}\sum_{i=1}^{p^{s-1}}\omega^{\Psi(x-y)_{i}} = 0$.
\end{corollary}



\begin{proof}
Since $x = fp^{s-1}$, by Lemma \ref{x-psi lemma} we have $\Psi(x) = f{\bf 1}$. Since $y = cp^{s} + dm \neq 0 \mod p^{s-1}$, by Lemma \ref{x-psi lemma} we have $\sum_{i=1}^{p^{s-1}}\omega^{\Psi(y)_{i}} = 0$. Complex conjugation is a field automorphism so it follows too that $\sum_{i=1}^{p^{s-1}}\omega^{-\Psi(y)_{i}} = 0$. It follows from Lemma \ref{Psi decomp} that $\Psi(x-y) = \Psi(x)-\Psi(y)+m(p-1){\bf 1}$, and so $\sum_{i=1}^{p^{s-1}}\omega^{\Psi(x-y)_{i}} = \omega^{f+m(p-1)}\sum_{i=1}^{p^{s-1}}\omega^{-\Psi(y)_{i}} = 0$.
\end{proof}

We will require the following result of Lam and Leung.

\begin{lemma}[Corollary 3.2, \cite{LamLeung}]\label{Lam}
If $\alpha_{1} + \cdots + \alpha_{r} = 0$ is a minimal vanishing sum of $n^{\rm th}$ roots of
unity, then after a suitable rotation, we may assume that all $\alpha_i$'s are $n_{0}^{\rm th}$ roots of
unity where $n_0$ is square-free.
\end{lemma}

The sum $\alpha_{1} + \cdots + \alpha_{r} = 0$ is \textit{minimal} if no proper subsums can be zero.  A rotation in this context is a multiplication of the sum by an $n^{\rm th}$ root of unity.

Suppose that for some multiset $X$ of elements of $\mathbb{Z}_{k}$, we have that $\sum_{x}\omega^{x} = 0$ is minimal, and further assume that each $\omega^x$ is an $n_{0}^{\rm th}$ root of unity for $n_{0}$ square-free. Then for each $x \in X$, $x = fp^{s-1}$ for some $f$. Lemma \ref{x-psi lemma} implies that $\sum_{x}\sum_{i=1}^{p^{s-1}}\omega^{\Psi(x)_{i}} = 0$, and then applying Corollary \ref{x-y corollary}, we get that $\sum_{x}\sum_{i=1}^{p^{s-1}}\omega^{\Psi(x-y)_{i}} = 0$ for all $y \neq 0 \mod p^{s-1}$. Any vanishing sum with terms that are not $n_{0}^{\rm th}$ roots of unity can only be scaled so that the terms are all $n_{0}^{\rm th}$ roots of unity by some $\omega^{y}$ where $y \neq 0 \mod p^{s-1}$. Thus we prove the following.

\begin{lemma}
If $\sum_{x}\omega^{x} = 0$ is minimal, then $\sum_{x}\sum_{i=1}^{p^{s-1}}\omega^{\Psi(x)_{i}} = 0$.
\end{lemma}
\begin{proof}
If the terms $\omega^{x}$ are $n_{0}^{\rm th}$ roots of unity then this is immediate from Lemma \ref{x-psi lemma}. Otherwise, we scale by some $\omega^{y}$ such that $y \neq 0 \mod p^{s-1}$ so that the terms are then $n_{0}^{\rm th}$ roots of unity. Then again we apply Lemma \ref{x-psi lemma} and prove the original equality using Corollary \ref{x-y corollary}.
\end{proof}

Finally, we can prove Theorem \ref{general morphism}.

\begin{proof}
Observe that the rows of $H^{\Psi}$ can be partitioned into $p^{s-1}$ blocks of size $n$ corresponding to the images of the rows of $L(H) + rmJ$ for $0 \leq r \leq p^{s-1}-1$. Given $H \in \BH(n,k)$, the Hermitian inner product of two distinct rows is zero. That is, for any two distinct rows ${\bf x} = [x_{1},\ldots,x_{n}]$ and ${\bf y} = [y_{1},\ldots,y_{n}]$ of $L(H)$, the Hermitian inner product of the corresponding rows of $H$ is of the form
\[
\sum_{i=1}^{n}\omega^{x_{i}-y_{i}} = 0.
\]
Since we can partition this equation into minimal sums, it follows that $\sum_{i=1}^{n}\sum_{j=1}^{p^{s-1}}\omega^{(\Psi(x_{i})-\Psi(y_{i}))_{j}} = 0$. That is, distinct rows of $H^{\Psi}$ from each block of $n$ rows are pairwise orthogonal. To see that two rows taken from distinct blocks are orthogonal, we observe that $tm \neq 0 \mod p^{s-1}$ for any $1 \leq t \leq p^{s-1}-1$, and so we also apply Corollary \ref{x-y corollary}.
\end{proof}

\begin{remark}
{\upshape The application of the map $\Psi_{2}$ to $H \in \BH(n,4)$ is equivalent to a familiar morphism $\BH(n,4) \rightarrow \BH(2n,2)$ of Turyn \cite{Turyn}. That is, for any $H \in \BH(n,4)$, the Hadamard matrix obtained from Turyn's morphism applied to $H$ is Hadamard equivalent to $H^{\Psi_{2}}$.}
\end{remark}

By Proposition \ref{isophi2} we know that $d^{\dag}({\bf x},{\bf y}) = d_{H}(\Psi({\bf x}),\Psi({\bf y}))$. We may also relate the minimum Hamming distance of $\Psi(\mathcal{C})$ directly to the minimum Hamming distance of $\mathcal{C}$, but less precisely.
\begin{proposition}
Let $\mathcal{C}$ be a $BH$-code of minimum Hamming distance $d$ obtained from a $\BH(n,p^{s}m$) with $p$ a prime not dividing $m$. Then the minimum distance $d'$ of $\Psi(\mathcal{C})$ is in the range $d(p-1)p^{s-2} \leq d' \leq dp^{s-1}$.
\end{proposition}
\begin{proof}
If $x_{i} \neq y_{i}$, then $p^{s-1}-p^{s-2} \leq d_{H}(\Psi(x_{i}),\Psi(y_{i})) \leq p^{s-1}$. Hence $d_{H}({\bf x},{\bf y})(p-1)p^{s-2} \leq d_{H}(\Psi({\bf x}),\Psi({\bf y})) \leq d_{H}({\bf x},{\bf y})p^{s-1}$.
\end{proof}
\begin{remark}
{\upshape The upper bound above is attainable. For example, the code $\mathcal{C}$ obtained from the Fourier matrix of order $27$ has minimum distance $18$. The code $\Psi(\mathcal{C})$ is a BH-code of length $243$, with minimum distance $162 = 18(3^2)$.}
\end{remark}

\section{Propelinear codes and cocyclic matrices}

The BH-matrix given in Example \ref{ejbhcnl}, $H$, is cocyclic over $\Z_8$ and its BH-code associated  $C_H$ is not linear. Can we define a propelinear structure in $C_H$?
Certainly, we can and this is not an isolated situation.

Let $G$ and $U$ be finite groups, with $U$ abelian, of orders $n$ and $k$, respectively. A map 
$\psi: G\times G\rightarrow U$ such that
\begin{equation}
\label{CocycleIdentity}
\psi(g,h)\psi(gh,k)=\psi(g,hk)\psi(h,k) \quad\forall \, g,h,k\in G
\end{equation}
is a \textit{cocycle} (\textit{over $G$}, \textit{with coefficients 
in $U$}). 
We may assume that $\psi$ is normalized, i.e., $\psi(g,1) = \psi(1,g)=1$ for all $g \in G$. 
For any (normalized) map $\phi: G\rightarrow U$, the cocycle $\partial\phi$  
defined by 
$\partial\phi(g,h)=\phi(g)^{-1}\phi(h)^{-1}\phi(gh)$ is
a \textit{ coboundary}. The set of all cocycles 
$\mbox{$\psi:\abk G\times G\rightarrow U$}$ forms an 
abelian group $Z^2(G,U)$ under pointwise multiplication. 
Factoring out the subgroup of coboundaries gives $H^2(G,U)$,
the \textit{second cohomology group of $G$ with coefficients in $U$}.

Given a group $G$ and $\psi\in Z^2(G, U)$, 
denote by $E_\psi$ the canonical 
central extension of $U$ by 
$G$; this has elements $\{(u,g) \mid u\in U,\; g\in G \}$ and 
multiplication $(u,g) \;(v,h)=(uv\hspace{.5pt} \psi(g,h),gh)$.
The image $U\times \{1\}$ of $U$ lies in the centre of $E_\psi$ and the set $T(\psi)=\{(1,g)\,\colon\, g\in G\}$ is a normalized transversal of $U\times \{1\}$ in $E_\psi$.
In the other direction, suppose that $E$ is a finite group with 
normalized transversal $T$ for a central subgroup 
$U$. Put
$G= E/U$ and 
$\sigma (tU) = t$ for $t\in T$. The map
$\psi_T:G\times G\rightarrow U$ defined by
$\psi_T(g,h)=\sigma(g)\sigma(h)\sigma(gh)^{-1}$ is a cocycle; 
furthermore, $E_{\psi_{T}}\cong E$.

Each cocycle $\psi\in Z^2(G,U)$ is displayed as a 
{\textit cocyclic matrix} $M_\psi$: under some indexing of 
the rows and 
columns by $G$, $M_\psi$ has entry $\psi(g,h)$ in 
position $(g,h)$.

 A $n\times n$ matrix $A=(a_{g,h})_{g,h\in G}$ is called \textit{G-invariant} (or just group invariant) if $a_{gk,hk}=a_{g,h}$ for all $g, h, k\in G$.
\begin{lemma}
If $A$ is $G$-invariant and   $a_{g,h}\in U$ then $\psi(g,h)=a_{g,0}^{-1}a_{g,h^{-1}} a_{0,h^{-1}}^{-1}$ is a cocycle.
\end{lemma}
\begin{remark}
{\upshape Every group invariant matrix with entries in $U$ is equivalent to a cocyclic matrix.}
\end{remark}

Fixing  $U=\langle \zeta_k\rangle$. A cocycle $\psi\in Z^2(G,\langle \zeta_k\rangle)$ is called \textit{ orthogonal} if, for each $g\neq 1\in G, \,\sum_{h\in G}\psi(g,h)=0$.  

\begin{proposition}\cite{Hor07}
$H_\psi\in \BH(n,k)$ if and only if $\psi\in Z^2(G,\langle \zeta_k\rangle)$ is orthogonal.
\end{proposition}

\noindent{\bf Fact:}
 A cocyclic Butson Hadamard matrix is not necessarily pairwise row and column balanced.

 
 \begin{proposition} \label{pfromctocg}
Given $\psi\in Z^2(G,\langle \zeta_k\rangle)$ and ${\bf x}=\zeta_k^\lambda\,[\psi(g,g_1),\ldots,\psi(g,g_n)]$ for a fixed order in $G=\{g_1=1,g_2,\ldots,g_{n}\}$.
Define $\pi_{\bf x}\in {\cal S}_n$ so that $\pi_{\bf x}^{-1}(j)=k$ where $g_k=g g_j$. Then 
\begin{enumerate}
    \item  ${\bf x}+ \pi_{\bf x}(y)=\zeta_k^{\lambda+\mu}\,\psi(h,g)\,[\psi(hg,g_1),\ldots,\psi(hg,g_n)]$   where $+$ means the componentwise product and ${\bf y}=\zeta_k^\mu\,[\psi(h,g_1),\ldots,\psi(h,g_n)]$. 
    \item $\pi_{{\bf x}+\pi_{\bf x}({\bf y})}=\pi_{\bf x}(\pi_{\bf y}).$
 
 \end{enumerate}

\end{proposition}
\begin{proof}
\begin{enumerate}
    \item Observe that $\pi_{\bf x}({\bf y}) = \zeta_{k}^{\mu}[\psi(h,gg_{1}),\ldots,\psi(h,gg_{n})]$. Hence the $i^{\rm th}$ component of ${\bf x} + \pi_{\bf x}({\bf y})$ is $\zeta_{k}^{\lambda + \mu}\psi(g,g_{i})\psi(h,gg_{i})$. Apply \eqref{CocycleIdentity} letting $(g,h,k) = (h,g,g_{i})$ and the result follows.
    \item Let ${\bf z}=\zeta_k^{\gamma}\,[\psi(\ell,g_1),\ldots,\psi(\ell,g_n)]$. From part 1 we know that ${\bf x} + \pi_{\bf x}({\bf y})$ is a scalar multiple of the $n$-tuple defined by $\psi(hg,-)$, and thus the $j^{\rm th}$ component of $\pi_{{\bf x}+\pi_{\bf x}({\bf y})}({\bf z})$ is $\psi(\ell,hgg_{j})$.  Now observe that the $k^{\rm th}$ component of $\pi_{{\bf y}}({\bf z})$ is $\psi(\ell,hg_{k})$. We have $\pi_{\bf x}(k) = j$ where $g_{k} = gg_{j}$, and thus the $j^{\rm th}$ component of $\pi_{\bf x}(\pi_{\bf y}({\bf z}))$ is $\psi(\ell,hg_{k}) = \psi(\ell,hgg_{j})$. 
\end{enumerate}
\end{proof}
\begin{corollary}\label{CMisBHFP}
Let $\psi\in Z^2(G,\langle \zeta_k\rangle)$ and $H_\psi\in \BH(n,k)$. Then the corresponding BH-code $C_H$ is a BHFP-code where
${\bf x}\star {\bf y}= {\bf x}+ \pi_{\bf x}({\bf y})$ for all ${\bf x},{\bf y} \in \mathcal{C}$.
\end{corollary}
\begin{proof}
Extend the definition of $\pi_{\bf x}$ for the rows ${\bf x}$ of $L(H_{\psi})$ to all of $C_{H}$ by letting $\pi_{{\bf x} + \alpha {\bf 1}} = \pi_{\bf x}$ for all $\alpha \in \Z_{k}$. The code ${C}_H$ is propelinear by Proposition \ref{pfromctocg}, and since ${\bf x}\star {\bf y}= {\bf x}+ \pi_{\bf x}({\bf y})$ for all ${\bf x},{\bf y} \in \mathcal{C}$, the first property of Definition \ref{FullPropDef} is satisfied. Finally observe that because $\pi_{\bf x}\in {\cal S}_n$ is defined so that $\pi_{\bf x}^{-1}(j)=k$ where $g_k=g g_j$, it follows that $\pi_{\bf x}$ fixes no coordinate when ${\bf x} \neq \alpha {\bf 1}$, and $\pi_{\alpha {\bf 1}} = Id_{\mathcal{S}_{n}}$ for all $\alpha \in \Z_{k}$.
\end{proof}

\begin{remark}
{\upshape A notorious class of cocyclic Butson matrices are those that are equivalent to group invariant (if $G$ is a cyclic group, they are called circulant Butson matrices). A construction method based on bilinear forms on finite abelian groups is given in \cite{DS19} which, in turn, provides BHFP-codes. Furthermore, for $G$ abelian it is known that Bent functions, group invariant generalized Hadamard matrices and abelian semiregular relative different sets are all either equivalent to group invariant Butson matrices or to group invariant Butson matrices with additional properties  (see \cite{Sch19}). Characterising group invariant Butson matrices in terms of BHFP codes is an open problem.} 
\end{remark}

We refer the reader to \cite[Section 3]{ABE19} for a detailed discussion on cocyclic generalized Hadamard matrices and the corresponding generalized Hadamard full propelinear codes. Rather than repeat this discussion, we note that the converse of Corollary \ref{CMisBHFP} holds under the assumption that the $\BH(n,k)$ is row and column balanced. A $\BH(n,p)$ is necessarily balanced, and is equivalent to a generalized Hadamard matrix over the cyclic group $C_{p}$ when $p$ is prime. 
\begin{corollary}
Let $C_H$ be a BHFP-code of length $n$ over $\Z_k$ coming from $H \in \BH(n,k)$, where $H$ is row and column balanced. Then $H$ is cocyclic.
\end{corollary}
\begin{proof}
The proof follows the proof of Proposition 4 and Corollary 2 of \cite{ABE19}.
\end{proof}

Let $H$ be a $\BH(n,k)$. We consider the following partition of its corresponding code. $C_{H} = \cup_{1 \leq \alpha \leq n}C_\alpha$ where $C_{\alpha} = \{[L(H)]_{\alpha} + \lambda {\bf 1}\}_{\lambda \in \Z_{k}}$ and $[L(D)]_i$ denotes the $i$-th row of $L(D)$.
\begin{example}\label{PropStructureEx}{\upshape
Let $H$ be the BH-matrix of Example \ref{ejbhcnl} since it is cocyclic over $\Z_8$. Then,

$$C_H=C_1\cup C_2\cup \ldots \cup C_8$$
can be endowed with a {\bf full propelinear structure} with the following group $\Pi$ of permutations
$$\pi_{\bf x}=\left\{
\begin{array}{cc}
I & \quad {\bf x}\in C_1\\
(1,2,3,4,5,6,7,8) &\quad {\bf x}\in C_2\\
(1,3,5,7) (2,4,6,8) &\quad  {\bf x}\in C_3\\
(1, 4, 7, 2, 5, 8, 3, 6)& \quad {\bf x}\in C_4\\
(1, 5)(2, 6)(3, 7)(4, 8)&\quad {\bf x}\in C_5\\
(1, 6, 3, 8, 5, 2, 7, 4)&\quad {\bf x}\in C_6\\
(1, 7, 5, 3)(2, 8, 6, 4)
&\quad {\bf x} \in C_7\\
(1, 8, 7, 6, 5, 4, 3, 2)&\quad {\bf x}\in C_8
\end{array}
\right.$$
$C_H$ is a BHFP-code with group structure $\Z_8 \times \Z_4$ and $\Pi\cong \Z_8$. The codewords are
\begin{align*}
C_1 &= \{ [0,0,0,0,0,0,0,0] + \lambda {\bf 1}\},\\
C_2 &= \{ [0,1,3,0,2,3,1,2] + \lambda {\bf 1}\},\\
C_3 &= \{ [0,3,2,1,0,3,2,1] + \lambda {\bf 1}\},\\
C_4 &= \{ [0,0,1,1,2,2,3,3] + \lambda {\bf 1}\},\\
C_5 &= \{ [0,2,0,2,0,2,0,2] + \lambda {\bf 1}\},\\
C_6 &= \{ [0,3,3,2,2,1,1,0] + \lambda {\bf 1}\},\\
C_7 &= \{ [0,1,2,3,0,1,2,3] + \lambda {\bf 1}\},\\
C_8 &= \{ [0,2,1,3,2,0,3,1] + \lambda {\bf 1}\}
\end{align*}
\noindent where $\lambda$ runs through $\Z_{4}$, and $C_H$ is a $(8,32,4)$-code over $\Z_4$.  $C_H$ has a group structure $\Z_8 \times \Z_4 \simeq \langle {\bf a}, {\bf 1} \mid {\bf a}^8= {\bf 1}^4=\zero \rangle$, where ${\bf a}=[0,1,3,0,2,3,1,2]$.}
\end{example}
An interesting family of BH-codes over $\Z_{p^s}$ are those associated to Kronecker products of Fourier matrices. They are denoted by ${\cal H}^{t_1,t_2,\ldots, t_s}$ (see Remark \ref{Defofaddcodes} and Proposition \ref{ACareBH}) and  since these matrices are cocyclic over $G=\Z_{p}^{t_s}\,\times \,\Z_{p^2}^{t_{s-1}}\,\times\,\ldots\,\times \,\Z_{p^{s-1}}^{t_2}\,\times\,\Z_{p^s}^{t_1-1}$, these codes can be endowed with a full propelinear structure by Corollary \ref{CMisBHFP} . Furthermore, for $p=2$ and $s=2$ in \cite{PR97}, it is shown that the image of ${\cal H}^{t_1,t_2}$ under the Gray map are in fact propelinear codes. 

\begin{example}
Considering ${\cal H}^{1,1,1}$, the $\Z_8$-additive code of length $n=8$ associated to $L(H)$ of Example \ref{ejemplosbhadd}.  Then, it can be endowed with a {\bf full propelinear structure} with the following group $\Pi$ of permutations $\Pi\cong \Z_2\times \Z_4$
 generated by $\pi_{\bf x}$ and $\pi_{\bf y}$ where
\[
\begin{array}{cc}
{\bf x} = [0,2,4,6,0,2,4,6], &
{\bf y} = [0,0,0,0,4,4,4,4],\\[2mm]
\pi_{\bf x} = (1,4,3,2)(5,8,7,6), &
\pi_{\bf y} = (1, 5)(2, 6)(3, 7)(4, 8).
\end{array}
\]
The full propelinear code is a group $(\mathcal{H}^{1,1,1},\star ) \cong \Z_{8}\times \Z_{4} \times \Z_{2} = \langle {\bf x}, {\bf y}, {\bf 1} \; | \; {\bf x}^8 = {\bf 0},~ {\bf y}^2 = {\bf 1}^4 = {\bf x}^4 \rangle$.

\end{example}

\section{Propelinear codes via the Gray map}

A natural question that arises is whether or not the generalized Gray preserves the property of being propelinear, or full propelinear. It is certainly true that the number of codewords in a BH-code $C$ obtained from $H$, a $\mathrm{BH}(n,mp^s)$, is the same as the number of codewords in the BH-code $C'$ obtained from $H^{\Psi}$. However, in general, it is not the case that $C'$ will be an isomorphic propelinear structure. A simple example to demonstrate this arises from the $\Z_{9}$-code $C$ obtained from the trivial $\mathrm{BH}(1,9)$, and the $\Z_{3}$-code $\Psi(C)$ obtained from the $\mathrm{BH}(3,3)$ matrix $H' = (1)^{\Psi}$ which written in log form is 
\[
L(H') = \left[
\begin{array}{ccc}
     0 & 0 & 0\\
     0 & 1 & 2\\
      0 & 2 & 1    
\end{array}
\right]
\]
The code $C$ is clearly linear, and as a group is isomorphic to the cyclic group $\Z_{9}$. It is also easily seen to be full propelinear by definition. However it is a short exercise to verify that $\Psi(C)$ cannot be both full propelinear and isomorphic to a cyclic group $G \cong \Z_{9}$ generated by any single element ${\bf x}$, no matter what the coordinate permutation $\pi_{\bf x}$ may be. The code $\Psi(C)$ does form a $2$-dimensional linear code (so it is also propelinear, but not full propelinear with ${\bf x} \star {\bf y} = {\bf x} + {\bf y}$ for all ${\bf x},{\bf y} \in \Psi(C)$), and $\Psi$ is a bijective map between codewords, but in general it is not always the case that $\Psi({\bf x} \star {\bf y}) = \Psi({\bf x}) \star' \Psi({\bf y})$ for any operation $\star'$, and as a consequence $\Psi$ will generally not preserve a group structure. The code $\Psi(C)$ of this example can also be with a full propelinear structure, but it will not be isomorphic as a group to $C$. It is generated by the codewords ${\bf x} = [0,1,2]$, and ${\bf 1}$, where $\pi_{\bf x} = (1,3,2)$. It is isomorphic to $\Z_{3}^2$.

However, we find that for the special case $\Psi_{2} : \mathbb{Z}_{4m} \rightarrow \mathbb{Z}_{2m}^{2}$, we can carefully construct an isomorphism between the groups of codewords $C$ and $C' = \Psi_{2}(C)$, and determine the group operation $\star'$ so that $(C,\star) \cong (C',\star')$. Let $\Psi = \Psi_{2}$.

\begin{theorem}\label{GrayGroup}
Let $m$ be an odd positive integer, and let $C \subseteq \Z_{4m}^{n}$ be a full propelinear code. Then the code $C'= \Psi(C)$ is full propelinear with group structure $(C',\star') \cong (C,\star)$.
\end{theorem}

\begin{proof}
First observe that $\Psi$ is a bijection from $C$ to $C'$, so we need to determine the group of permutations for $C'$ and show that $\Psi : (C,\star) \rightarrow (C',\star')$ is a homomorphism.  We start with the $n=1$ case, so we just need to show that we can choose $\rho_{x} \in \mathcal{S}_{2}$ for each $x \in \Z_{4m}$ so that $\Psi(x) + \rho_{x}(\Psi(y)) = \Psi(x+y)$ for all $y$. We will see that $\rho_{x} = (1,2)^x$, i.e., $\rho_{x}$ permutes the two coordinates of a word in $\Z_{2m}^{2}$ or not, according to the parity of $x$. We adhere to the notation of the proof of Lemma \ref{Psi decomp}. Fix $x = 4a+mb$ and let $y = 4c+md$ where $0 \leq b,d \leq 3$, so $x+y = 4(a+c) + m(b+d)$ with the value of $b+d$ taken modulo $4$. A complete proof requires a verification that $\Psi(x) + \rho_{x}(\Psi(y)) = \Psi(x+y)$ for each pair $(b,d) \in \Z_4$, but for brevity we take $(b,d) = (3,1)$ as an example and leave the rest to the reader. Observe that 
\begin{align*}
\Psi(x) &= [2a,2a]+m\Phi(3) = [2a,2a]+m([0,1] + [1,1]) = [2a+m,2a],\\
\Psi(y) &= [2c,2c]+m\Phi(1) = [2c,2c]+m([0,1] + [0,0]) = [2c,2c+m],\\
\Psi(x+y) &= [2(a+c),2(a+c)] + m\Phi(0) = [2(a+c),2(a+c)].
\end{align*}
Since $b=3$, $x$ is odd, and so $\rho_{x} = (1,2)$. It follows that $\Psi(x)+\rho_{x}(\Psi(y)) = \Psi(x+y)$. This verifies the $1$-dimensional case.

Now suppose that $C$ is full propelinear of length $n$, and let ${\bf x},{\bf y} \in C$, with ${\bf x} \star {\bf y} = {\bf x} + \pi_{\bf x}({\bf y})$. Let $\pi_{\Phi({\bf x})} \in \mathcal{S}_{2n}$ permute the $n$ blocks of size $2$, labelled $b_{1},\ldots,b_{n}$, according to the action of $\pi_{{\bf x}}$ on a word of length $n$. That is, $\pi_{\Phi({\bf x})}(b_{i}) = b_{j}$ if and only if $\pi_{\bf x}(i) = j$. Then $\pi_{\Phi({\bf x})}(\Psi({\bf y})) = \Psi(\pi_{\bf x}({\bf y}))$. 
Further, let $\rho_{i} = (2i-1,2i)$ be the permutation swapping the entries of the block $b_{i}$, and write $\rho_{{\bf x}} = \prod_{i=1}^{n}\rho_{i}^{x_{i}}$. It follows that $\Psi({\bf x})\star'\Psi({\bf y}) := \Psi({\bf x}) + \rho_{{\bf x}}\pi_{\Phi({\bf x})}(\Psi({\bf y})) = \Psi({\bf x}+\pi_{{\bf x}}({\bf y})) = \Psi({\bf x} \star {\bf y})$. Thus $\Psi$ is a bijective homomorphism from $(C,\star)$ to $(C',\star')$.

It remains to verify that the permutation $\rho_{{\bf x}}\pi_{\Phi({\bf x})} = Id_{\mathcal{S}_{2n}}$ whenever $\Psi({\bf x}) = \alpha{\bf 1}_{2n}$ for any $\alpha \in \Z_{2m}$, and has no fixed coordinate otherwise. Let $S = C \cap \{\alpha {\bf 1}_{n} \; : \; 0 \leq \alpha \leq 4m-1 \}$ and let $X \subset S$ be the subset $X = C \cap \{2\alpha {\bf 1}_{n} \; : \; 0 \leq \alpha \leq 2m-1 \}$. Note first that $\Psi(X)$ is the set $X' = C' \cap \{\alpha {\bf 1}_{2n} \; : \; 0 \leq \alpha \leq 2m-1 \}$. It is clear that $\rho_{{\bf x}}\pi_{\Phi({\bf x})} = Id_{\mathcal{S}_{2n}}$ for all ${\bf x} \in X$. Further, for any ${\bf s} \in S \setminus X$, $\rho_{{\bf s}} = (1,2)(3,4)\cdots (2n-1,2n)$, and so does not fix any coordinate. Finally, for any codeword ${\bf c} \in C \setminus S$, $\pi_{{\bf c}}$ does not fix any coordinate of $\Z_{4m}^{n}$, and it follows that $\pi_{\Phi({\bf c})}$ does not fix any coordinate of $\Z_{2m}^{2n}$.
\end{proof}
\begin{corollary}\label{GrayGroupCor}
Let $m$ be an odd positive integer, and let $H \in \mathrm{BH}(n,4m)$. If the BH-code $C$ obtained from $H$ is full propelinear with group structure $G$, then the BH-code $C'$ obtained from $H^{\Psi}$ where $\Psi$ is full propelinear with group structure $G' \cong G$.
\end{corollary}


\begin{example}{\upshape
Let ${\cal H}^{3,0}$ be the BH-code associated to $F_4\otimes F_4\in \BH(16,4)$ and  $H^{3,0}$ be its image by the Gray map which is known to be a nonlinear code (see \cite[Table 1]{FVV19}). ${\cal H}^{3,0}$ is full propelinear, with permutation group $\Pi \cong \Z_{4}^{2}$ generated by $\pi_{\bf x}$ and $\pi_{\bf y}$ where
\[
\begin{array}{cc}
{\bf x} = [0,1,2,3,0,1,2,3,0,1,2,3,0,1,2,3],\\[2mm]
{\bf y} = [0,0,0,0,1,1,1,1,2,2,2,2,3,3,3,3],\\[2mm]
\pi_{\bf x} = (1,4,3,2)(5,8,7,6)(9,12,11,10)(13,16,15,14),\\[2mm]
\pi_{\bf y} = (1, 13,9,5)(2, 14,10,6)(3, 15,11,7)(4,16,12,8).
\end{array}
\]
The corresponding permutations $\rho_{\bf x}\pi_{\Phi({\bf x})}$ and $\rho_{\bf y}\pi_{\Phi({\bf y})}$ are as follows:
\[
\begin{array}{cc}
\rho_{\bf x}\pi_{\Phi({\bf x})} = (1,7,6,4)(2,8,5,3)(9,15,14,12)(10,16,13,11) & \\ (17,23,22,20)(18,24,21,19)(25,31,30,28)(26,32,29,27),\\[2mm]
\rho_{\bf y}\pi_{\Phi({\bf y})} = (1, 25,17,9)(2, 26,18,10)(3, 28,19,12)(4,27,20,11) & \\ 
(5,29,21,13)(6,30,22,14)(7,32,23,16)(8,31,24,15).
\end{array}
\]
Thus, $H^{3,0}$ can be endowed with a {\bf full propelinear structure} 
with the group $\langle \rho_{\bf x}\pi_{\Phi({\bf x})},\rho_{\bf y}\pi_{\Phi({\bf y})} \rangle$ of permutations, which is non-abelian of order $32$. This group contains the element $(\rho_{\bf x}\pi_{\Phi({\bf x})})(\rho_{\bf y}\pi_{\Phi({\bf y})})(\rho_{\bf x}\pi_{\Phi({\bf x})})^{-1}(\rho_{\bf y}\pi_{\Phi({\bf y})})^{-1} = \rho_{\bf 1}\pi_{\Phi({\bf 1})} = (1,2)(3,4)\cdots (31,32)$. The groups $(\mathcal{H}^{3,0},\star) \cong (H^{3,0},\star')$ are isomorphic to $\Z_{2}\times\Z_{4}\times\Z_{8}$. }
\end{example}

\begin{remark}{\upshape
Even though the codes $C$ and $C'$ are isomorphic as groups according to Theorem \ref{GrayGroup}, the example above shows that the underlying groups of coordinate permutations are not necessarily isomorphic. As a simpler example, take the trivial $1$-dimensional $\Z_{4}$ code and its image in $\Z_{2}^2$. Here, $\Psi : [0],[1],[2],[3] \mapsto [0,0],[0,1],[1,1],[1,0]$. Both are cyclic, generated by $[1]$ and $[0,1]$ respectively, but the group of coordinate permutations of $\Z_4$ is necessarily trivial, and the group of coordinate permutations of the image is generated by $\rho_{[1]}\pi_{[0,1]} = (1,2)$. More generally, if $C$ is a BHFP-code obtained from a $\mathrm{BH}(n,4m)$ with group $\Pi$ of coordinate permutations then by Definition \ref{FullPropDef}, $|\Pi| = n$, and the group of coordinate permutations for $\Psi(C)$ will be of order $|\Pi'| = 2n$.}
\end{remark}
\section*{Acknowledgements}
The authors would also like to thank  Kristeen Cheng for her reading of this manuscript.
The first author was supported by 
the project FQM-016 funded by JJAA (Spain). The second author was supported by the Spanish grant TIN2016-77918-P (AEI/FEDER, UE). The third author was supported by the Irish Research Council (Government of Ireland Postdoctoral Fellowship, GOIPD/2018/304).

\newpage


\end{document}